\documentclass[twoside]{article}
\usepackage[a4paper]{geometry}
\usepackage[latin1]{inputenc} 
\usepackage[T1]{fontenc} 
\usepackage{RR}
\usepackage{hyperref}
\usepackage{amsmath,amssymb,amsthm,anysize}
\usepackage{tikz}
\usetikzlibrary{arrows}
\usepackage{algorithm}
\usepackage[noend]{algorithmic}

\usepackage{textcomp}
\newtheorem{theorem}{Theorem}[section]

\newtheorem{example}[theorem]{Example}
\newtheorem{remark}[theorem]{Remark}

\newtheorem{definition}[theorem]{Definition}
\usepackage{graphicx}
\usepackage{epstopdf}
\RRNo{8716}
\RRdate{April 2015}
\RRauthor{
Konstantin Avrachenkov\thanks{Inria Sophia Antipolis, France, konstantin.avratchenkov@inria.fr}%
\and Vikas Vikram Singh \thanks{Inria Sophia Antipolis, France, vikasstar@gmail.com}
\thanks[sfn]{The authors' names are given in alphabetical order.}
}
\authorhead{K. Avrachenkov  \& V. V. Singh}
\RRtitle{La dynamique de meilleure r\'eponse de coalitions et l'\'equilibre de Nash fort}
\RRetitle{Stochastic Coalitional Better-response Dynamics
 and  Strong Nash Equilibrium}
\titlehead{Stochastic Coalitional Better-response Dynamics
 and  Strong Nash Equilibrium}
 \RRresume{
 Nous consid\'erons un processus de formation de coalitions entre les joueurs
 d'un jeu fini strat\'egique sur l'horizon de temps infini.
 \'A chaque \'etape, une coalition form\'ee au hasard fait une d\'eviation conjointe
 de l'ensemble actuel des actions de telle sorte qu'au nouveau ensemble des actions,
 tous les joueurs de la coalition sont strictement b\'en\'efici\'e. Telles d\'eviations
 d\'efinissent une dynamique de meilleure r\'eponse de coalitions, Coalitional
 Better-Response dynamics en anglais (CBR), qui est en g\'en\'eral stochastique.
 La dynamique CBR soit converge vers un \'equilibre
 de Nash fort ou \'a un cycle ferm\'e. En outre, nous supposons que \'a chaque \'etape
 une coalition s\'electionn\'ee fait une faute avec faible probabilit\'e qui ajoutent
 des mutations (perturbations) dans la dynamique CBR. Nous prouvons que
 tous les \'equilibres de Nash forts et les cycles ferm\'es sont stochastiquement
 stable, ce est \'a dire, ils sont choisis par CBR perturb\'ee quand les mutations
 disparaissent. Une affirmation similaire a lieu pour l'\'equilibre de Nash
 fort et stricte. Nous appliquons la dynamique CBR aux jeux de formation
 de r\'eseau et nous prouvons que tous les r\'eseaux fortement stables et des
cycles ferm\'es sont stochastiquement stable.
}
\RRabstract{
 We consider coalition formation among players in an $n$-player finite strategic game over infinite horizon.
At each time a randomly formed coalition makes a joint deviation from a current action profile such that
at new action profile all players from the coalition are strictly benefited.
Such deviations define a  coalitional better-response (CBR) dynamics that is in general stochastic.
The CBR dynamics  either converges  to  a strong Nash equilibrium or stucks in a closed cycle.
We also assume that at each time a selected coalition makes mistake in deviation with small probability
that add mutations (perturbations) into CBR dynamics. We prove that all strong Nash equilibria and
closed cycles are stochastically stable, i.e., they are selected by perturbed CBR dynamics as mutations
vanish. Similar statement holds for strict strong Nash equilibrium. We apply CBR dynamics to the network
formation games and we prove that all strongly stable networks and closed cycles are stochastically stable.
 }
\RRmotcle{Forte \'equilibre de Nash,
Coalitionnelle meilleure r\'eponse,
Stabilit\'e stochastique,
Jeux de formation de r\'eseau,
R\'eseaux fortement stable.}

\RRkeyword{ Strong Nash equilibrium, Coalitional better-response,  Stochastic stability,
  Network formation games, Strongly stable networks.}
\RRprojets{Maestro}
\RCSophia 

\begin{document}
\makeRR   
\section{Introduction}
 Nash equilibrium is the most desirable solution concept  in non-cooperative game theory.
 When a strategic game is played  repeatedly over infinite horizon then the Nash equilibrium that is played in the long run
 depends on an initial action profile as well as the way all the players choose their actions at each time.
 Several discrete time dynamics have been studied in the literature to study the  Nash
 equilibrium selection in the long run.
 Young \cite{Young}  considered an  $n$ players strategic game where at each time all the  players make
a simultaneous move and each player chooses an action that is the
best response to  $k$  previous games among the $m$, $k\le m$,  most recent games in past.
In general this dynamics need not converge to a Nash equilibrium, it may stuck into a closed cycle.
Young  \cite{Young} also considered the case where at each time  with  small probability
each player makes mistake and chooses some non-optimal action. These mistakes add  mutations into the dynamics.
In general the mutations can be sufficiently small which leads to the definition of stability of Nash equilibrium as mutations vanish.
This type of stability  is known as stochastic stability.
Young proposed an algorithm to compute the stochastically  stable  Nash equilibria.
For $2\times 2$ coordination games he showed that the  risk dominant Nash equilibrium is stochastically stable.
Kandori et al. \cite{KMR} considered a different dynamic model where
 at each time each player plays with every other player in pairwise contest. The pairwise contest is given by
$2\times 2$ symmetric matrix game and each player chooses an action which has higher expected average payoff. The mutations
are present into dynamics due to wrong actions taken by the players.
 For $2\times 2$ coordination games they showed that a risk dominant
Nash equilibrium is stochastically stable. That is, for $2\times 2$ coordination games  the dynamics given by Young \cite{Young}
and Kandori et al. \cite{KMR}
selects the same Nash equilibrium.
 Fudenberg et al. \cite{Fudenberg} proposed a dynamics where at each time only one player is selected to choose actions.
 The mutations with small probability
 also occur  at each time. The risk dominant Nash equilibrium
in $2\times 2$ coordination games need not be stochastically stable under this dynamics.

The Nash equilibrium concept is inadequate for the situations where players can a priori communicate, being in a position to
form a coalition and jointly deviate in a coordinated way.
To capture such situations the strong Nash equilibrium (SNE) introduced by Aumann \cite{Aumann} is
an adequate solution concept.
From an SNE there is no coalition that can deviate to a new action profile
such that at new action profile the actions of all players from outside of the coalition are same
as at SNE and all the players from the coalition are strictly benefited.
There is another equilibrium notion that is stronger than  the SNE. Such equilibrium is called as
strict strong Nash equilibrium (SSNE). From an SSNE there is no coalition that can
 deviate to a new action profile  such  that at new action profile the actions of all players from outside of the coalition are same
as at SSNE and all the players from coalition get at least as much as at SSNE and
at least one player is strictly benefited. It is clear that an SSNE is always an SNE.
As motivated from the application of SNE in network formation games by Dutta and Mutuswami \cite{Dutta-Mutu} and
SSNE in network formation games by Jackson and van den Nouweland
\cite{Jack-Nou}, Jackson \cite{Jackson2} we restrict ourselves to only pure actions.
A network that is stable against the deviations of all coalitions is called as strongly stable network
and under top convexity condition on payoff functions it indeed exists as shown by Jackson and van den Nouweland
\cite{Jack-Nou}.
An SNE need not always exist and in such case there exists
some set of action profiles forming a closed cycle such that it is possible to reach from one action profile to another via
sequence of  improving deviations from the coalitions; and it is not possible to reach an action profile outside of the closed cycle
from an action profile belonging to closed cycle
via  improving deviations from the coalitions.

There are many dynamics for equilibrium selection in the literature, as discussed before, describing various situations of the dynamic play.
 To the best of our knowledge so far  no dynamics has been proposed that captures the situation where  at each time
players are allowed to form a coalition  and make a move in a coordinated way.
In this paper we propose a CBR dynamics  where at each time players are allowed to form a coalition and  make
a joint deviation from the current action profile if it is strictly beneficial for all the members of the coalition.
We assume that the coalition formation is random and at each time only one coalition can be formed.
We also consider the situation where at each time the formed  coalition
makes wrong decision with small probability, i.e., they make a move to an action profile
where all the players from the coalition are not strictly benefited. These mistakes work as mutations and add perturbations
into CBR dynamics.
We prove that the perturbed CBR dynamics selects all strong Nash equilibria and closed cycles in the long run as mutations vanish,
i.e., all strong Nash equilibria and closed cycles are stochastically stable.
For $2\times 2$ symmetric coordination games this dynamics always selects a payoff dominant Nash equilibrium instead of
risk dominant Nash equilibrium because a payoff dominant Nash equilibrium is an SNE.
The similar CBR dynamics can be given for the
case where each time a coalition deviate from a current action profile such that all  players from the coalition are at least
as well off
at new action profile  and at least one player is strictly better off.
Under such dynamics all strict strong Nash equilibria
and closed cycles are stochastically stable.

We apply CBR dynamics corresponding to SSNE to network formation games where nodes (players) of a network form a coalition and make a move
to a new network if it offers each player at least as much as it is in the current network and at least one player  gets strictly
better payoff. The mutations are present due to the wrong decisions taken by the coalitions.
We prove that all strongly stable networks and closed cycles are stochastically
stable.

The paper is organized as follows. Section \ref{model} contains the model and few  definitions.
We describe the CBR dynamics in Section \ref{Dynamicplay}. Section \ref{net-games}
contains the application of CBR dynamics to network formation games. We conclude our paper in Section \ref{conclusion}.
As a by-product, we give an algorithm
to compute an SNE in Appendix~\ref{algo}.

 \section{The Model}\label{model}
 We consider an $n$-player strategic game whose components are defined as follows:

\begin{enumerate}
 \item $N=\{1,2,\cdots,n\}$ is a finite set of players.

 \item $A_i$ is a finite set of actions of player $i$ and its  element is denoted by $a_i$. We denote $A=\prod_{i=1}^n A_i$ as a
 set of all action profiles and  $a=(a_1,a_2,\cdots,a_n)$ denotes an element of $A$. Let $\mathcal{S}$ be
 the set of all coalitions among players.
 For a coalition $S\in \mathcal{S}$, define
 $A_S=\prod_{i\in S} A_i$ whose element  is denoted by $a_S$ and $a_{-S}$ denotes an action profile of players outside $S$.

 \item $u_i:A\rightarrow \mathbb{R}$ is a payoff  function of player $i$. Specifically, player $i$ receives payoff
 $u_i(a_1,a_2,\cdots,a_n)$
 when  each player $i$, $i=1,2,\cdots,n$, chooses action $a_i$.
\end{enumerate}
In non-cooperative  games, the Nash equilibrium is stable against unilateral deviations, i.e.,  no player has an incentive
to deviate unilaterally from it.
But, the  Nash equilibrium fails to capture the situation where a priori the players can communicate with each other. In such cases
 some of the players can form a coalition
and jointly deviate from a current action profile
 if at new action profile  each player from the coalition is strictly benefited. In some cases players also make a joint deviation
 from a current action profile if at new action profile all the players of coalition are at least as well off and at least one player
 is strictly better off. Such deviations lead to the definitions of strong Nash equilibrium \cite{Aumann} and strict strong
 Nash equilibrium which we define next. As motivated from the application of SNE in network formation games by Dutta and
 Mutuswami \cite{Dutta-Mutu} and
 application of SSNE in network formation games by Jackson and van den Nouweland
\cite{Jack-Nou},  Jackson \cite{Jackson2} we restrict ourselves to
 pure actions.

 \begin{definition}[Strong Nash Equilibrium]\label{SNE}
  An action profile $a^*$ is said to be a strong Nash equilibrium if there is no $S\in \mathcal{S}$ and $a\in A$ such that
  \begin{enumerate}
   \item $a_i=a_i^*, \ \forall \ i\notin S$.
   \item $u_i(a)> u_i(a^*), \ \forall \ i\in S$.
  \end{enumerate}
 \end{definition}
 Let $A(S,a)$ be the set of all  action profiles reachable
from $a$ via deviation of coalition $S$. It  is defined as,
$$A(S,a)=\{a'|a'_i=a_i, \ \forall \ i\notin S \ \mbox{and} \ a'_i\in A_i, \ \forall \ i\in S\}.$$
A coalition always has option to do nothing, so $a\in A(S,a)$.
Let $\mathcal{I}_1(S,a)$ be the set of improved action profiles reachable from an action profile $a$ via deviation of coalition $S$,~i.e.,
\begin{equation}\label{Impr_SNE}
\mathcal{I}_1(S,a)=\{a'|a'_i=a_i, \ \forall \ i\notin S \ \mbox{and} \ u_i(a')> u_i(a), \ \forall \ i\in S \}.
\end{equation}
For an improved action profile $a'\in \mathcal{I}_1(S,a)$, an action profile $a'_S$ of all players from $S$ is called as
a \mbox{better-response}
of coalition $S$ against a fixed action profile $a_{-S}$ of the players outside $S$.
Define, $\overline{\mathcal{I}}_1(S,a)=A(S,a)\setminus \mathcal{I}_1(S,a)$ as a set of all action profiles due to the erroneous
decisions of coalition $S$.
It is clear that $a\notin \mathcal{I}_1(S,a)$,
so $a\in \overline{\mathcal{I}}_1(S,a)$. That is, $\overline{\mathcal{I}}_1(S,a)$ is always nonempty for all $S$ and $a$.
An SNE need not always exist. In such a case there exists some set of action profiles lying on a closed cycle
and all such action profiles  can be reached from each other via an improving path. The definitions of
closed cycle and improving path are as follows:

\begin{definition}[Improving Path]\label{impr-path}
An improving path from $a$ to $a'$ is a sequence of action profiles and coalitions $a^1, S_1, a^2,  \cdots, a^{m-1}, S_{m-1}, a^m$ such that
$a^1=a$, $a^m=a'$ and \mbox{$a^{k+1}\in \mathcal{I}_1(S_k,a^k)$} for all $k=1,2,\cdots,m-1$.
\end{definition}

\begin{definition}[Cycles]\label{CC}
 A set of action profiles $C$ form a  cycle if for any $a\in C$ and $a'\in C$
 there exists  an improving path connecting $a$ and $a'$.
 A cycle is said to be a closed cycle if no action profile in $C$ lies on an improving
 path leading to an action profile that is not in $C$.
\end{definition}

\begin{theorem}\label{SNE_thm}
  There always exists a strong Nash equilibrium or a closed cycle of action profiles.
\end{theorem}
\begin{proof}
An action profile is an SNE if and only if it is not possible for any  coalition to make an improving deviation
from it to another action profile.
 So, start at an action profile. Either it is SNE or there exists a coalition that can make an improving deviation
to another action profile. In the first case result is established. For the second case the same thing holds, i.e., either this
 new action profile is an SNE or there exists a coalition that can make an improving deviation to another action
 profile. Given the finite number of action profiles, the above process either finds an action profile which is an
 SNE
 or it reaches to the starting action profile, i.e., there exists a cycle. Thus, we have proved that there always exists  either
 an SNE or a cycle. Suppose there are no strong Nash equilibria.
 Given the finite number of action profile and non-existence of strong Nash equilibria there must
 exists a maximal set $C$ of action profiles such that for any $a\in C$ and $a'\in C$
 there exists  an improving path connecting $a$ and $a'$ and  no action profile in $C$ lies on an improving
 path leading to an action profile that is not in $C$. Such a set $C$ is a closed cycle.
\end{proof}
An SSNE can be defined similarly.
An action profile $a^*$ in Definition \ref{SNE} is said to be SSNE
 if the condition 1 is same and  the condition 2 is
 $u_i(a)\ge u_i(a^*)$ for all $i\in S$ with at least one strict inequality. That is, $a^*$ is an SSNE if it is
 not possible for any coalition $S\in \mathcal{S}$
 to deviate from $a^*$ to some $a\in A$ such that the actions of all players outside $S$
 are same in both $a$ and $a^*$ and at $a$ all players from $S$ are at least as well off as  at $a^*$  and
 the payoff of at least one player at $a$ is better than at $a^*$.
 In this case, for the given action profile  $a$ and  coalition $S$ the set of improved action profiles
 $\mathcal{I}_2(S,a)$
 is defined~as,
 \begin{equation}\label{Impr_SSNE}
  \mathcal{I}_2(S,a)=\{a'|a_i'=a_i, \ \forall \ i\in S, \ \mbox{and}\ u_i(a')\ge u_i(a),\ \forall \ i\in S,
u_j(a')>u_j(a), \ \mbox{for some} \ j\in S\}.
 \end{equation}
 and
 $\overline{\mathcal{I}}_2(S,a)=A(S,a)\setminus \mathcal{I}_2(S,a)$. The definitions of improving path and cycles can be defined
 analogously to previous case. A result similar to Theorem \ref{SNE_thm} holds, i.e., there always exists at least
 an SSNE or a closed cycle of action profiles. An SSNE is always an SNE, i.e., the set of strict strong Nash equilibria
 is a subset of the set of strong Nash equilibria.
 An SNE is a weakly Pareto optimal
 and an SSNE is a Pareto optimal.
 Now, we give few examples illustrating the presence of SNE, SSNE and closed cycle.

\begin{example}
Consider a two player game
 $$
 \begin{matrix}
\  b_1 \  ~~~~~~~~~~~~~b_2 \\
\begin{matrix} a_1\\a_2\end{matrix}
\begin{pmatrix}
  (-2,-2) & (-10,-1)
   \vspace{.15cm} \\
  (-1,-10) & (-5,-5)
\end{pmatrix}.
\end{matrix}
 $$
\end{example}
\noindent The above game represents a famous example of prisoner's dilemma. Here $(a_2,b_2)$ is the only  Nash equilibrium
that is not an SNE because both player can jointly deviate to $(a_1,b_1)$ where both of them are strictly better off.
So, in this game there is no SNE and SSNE. The closed cycle of action profile is given in Figure \ref{Cl-cyc1}

\begin{figure}[htb]
\begin{center}
\begin{tikzpicture}[->,>=stealth',shorten >=1pt,auto,node distance=3.2cm,
 thick,main node/.style={circle,fill=blue!20,draw,font=\footnotesize}]
  \node[main node] (1) {$(a_1,b_1)$};
  \node[main node] (2) [below left of=1] {$(a_2,b_1)$};
  \node[main node] (3) [below right of=2] {$(a_2,b_2)$};
  \node[main node] (4) [below right of=1] {$(a_1,b_2)$};

  \path[every node/.style={font=\footnotesize}]
    (1) edge node [right] {\{2\}} (4)
        edge node[left] {\{1\}} (2)
    (2) edge node [left] {\{2\}} (3)
    (3) edge node [right] {\{1,2\}} (1)
    (4) edge node [right] {\{1\}} (3);
\end{tikzpicture}
\caption{Closed Cycle}\label{Cl-cyc1}
\end{center}
\end{figure}
\noindent A directed edge $(a_2,b_2)\xrightarrow{\{1,2\}} (a_1,b_1)$ of Figure \ref{Cl-cyc1}
represents a deviation by coalition $\{1,2\}$. The other
directed edges of the closed cycle are similarly defined.


\begin{example}\label{ex1}
Consider a two player game
 $$
 \begin{matrix}
\  b_1 \  ~~~~~b_2 \  ~~~~~b_3\\
\begin{matrix} a_1\\a_2\\ a_3\end{matrix}

 \begin{pmatrix}
  (4,4) & (0,0) & (0,0)
   \vspace{.1cm} \\
  (0,0) & (4,5) & (1,6)
  \vspace{.1cm}\\
  (0,0) & (2,5) & (6,1)
 \end{pmatrix}.
 \end{matrix}
 $$
\end{example}
\noindent This example has both SNE and closed cycle. The action profile  $(a_1,b_1)$ is an SNE and the closed cycle is defined as below:

\begin{figure}[htb]
\begin{center}
\begin{tikzpicture}[->,>=stealth',shorten >=1pt,auto,node distance=3.2cm,
 thick,main node/.style={circle,fill=blue!20,draw,font=\footnotesize}]

  \node[main node] (1) {$(a_2,b_2)$};
  \node[main node] (2) [below right of=1] {$(a_2,b_3)$};
  \node[main node] (3) [below left of=2] {$(a_3,b_3)$};
  \node[main node] (4) [below left of=1] {$(a_3,b_2)$};

  \path[every node/.style={font=\footnotesize}]
    (1) edge node [right] {\{2\}} (2)
    (2) edge node [right] {\{1\}} (3)
    (3) edge node [left] {\{2\}} (4)
    (4) edge node [left] {\{1\}} (1);
\end{tikzpicture}
\caption{Closed Cycle}\label{cl-cyc2}
\end{center}
\end{figure}
\noindent But, $(a_1,b_1)$ is not an SSNE because according to the  improved action profile set defined by \eqref{Impr_SSNE},
both player can make a joint deviation from  action
profile $(a_1,b_1)$ to $(a_2,b_2)$. But, if we change the payoff vector corresponding to $(a_2,b_2)$ from
$(4,5)$ to $(4-\alpha,5)$ for $\alpha>0$ then $(a_1,b_1)$ is also an SSNE.

\section{Dynamic play} \label{Dynamicplay}
We consider the situation where $n$ players play the strategic game defined in Section~\ref{model}.
We assume that the players can a priori communicate  with each other and hence they can form a coalition
and jointly deviate from the current action profile to a new action profile if new action profile
is strictly beneficial for all members of coalition.
We consider the coalition formation over infinite horizon.
That is, at each time a coalition is randomly formed and it makes a deviation from current action profile
to a new action profile
such that at new action profile
the actions of the players outside the coalition remain same as before and
each player of the coalition
is strictly benefited. If there is no such improved action profile for a coalition then it does not deviate.
The same thing repeats at next stage and it continues for infinite horizon.
Such deviations define a coalitional better-response (CBR) dynamics.
We assume that the coalition formation is random and at each time only one coalition can be formed. If there are more than one improved
action profiles for a coalition then each improved action profile can be chosen with positive probability.
That is, the CBR dynamics is stochastic.
The CBR dynamics defines a Markov chain over a finite set of action profiles $A$.
We also assume that at each time  selected coalition makes mistake and make a joint deviation
to an action profile where all members
of the coalition are not strictly benefited. This happens with very small probability. Such mistakes add mutations into CBR dynamics.
The mutations add another level of stochasticity in the CBR dynamics and as a result
we have perturbed Markov chain, see e.g., \cite{AFH02,AFH13}.
We are interested in the action profile which is going to be selected by the CBR dynamics as mutations vanish.
We next describe the stochastic CBR dynamics as discussed above.

\subsection{A stochastic CBR dynamics  without mistakes}\label{Stoc_DP}
At each time $t=0,1,2,\cdots$ a coalition $S_t$ is selected randomly with probability
$p_{S_t}>~0$.
We assume that at each time  selected coalition  makes an improving deviation from current action profile $a^t$,
i.e., at time $t+1$, the new action profile is
$a^{t+1}\in \mathcal{I}_1(S_t,a^t)$ with probability
$p_{\mathcal{I}_1}(a^{t+1}|S_t,a^t)$ where $p_{\mathcal{I}_1}(\cdot|S_t,a^t)$ is a probability measure over finite set
$\mathcal{I}_1(S_t,a^t)$.
When there are no improving deviations for coalition  $S_t$   then $a^{t+1}=a^t$. Let $X_t^0$ denotes the action profile
at time $t$, then $\{X^0_t\}_{t=0}^\infty$ is a finite Markov chain on set $A$.  The transition law $P^0$ of the Markov chain
 is defined as follows:
\begin{align}\label{TPM_Static_unperturb}
 P^0(X_{t+1}^0=a'|X_t^0=a)&=\sum_{S\in \mathcal{S}; \mathcal{I}_1(S,a)\neq \phi}p_S ~ p_{\mathcal{I}_1}(a'|S,a)1_{\mathcal{I}_1(S,a)}(a')
 +\sum_{S\in \mathcal{S}; \mathcal{I}_1(S,a)= \phi} p_{S}1_{\{a'=a\}}(a'),
\end{align}
where $1_{B}$ is an indicator function for a given set $B$. It is clear that the strong Nash equilibria and closed cycles
are the recurrent classes of $P^0$. An SNE corresponds to an absorbing state of $P^0$ and
a closed cycle corresponds to a recurrent class of $P^0$ having more than one action profiles.

From Example \ref{ex1} it is clear that in general the closed cycles together with strong Nash equilibria can be present in a game.
   In that case the CBR
 dynamics need not converge. In Example \ref{ex1}
the  CBR dynamics need not converge to SNE $(a_1,b_1)$ because once CBR dynamics enter into  closed
cycle given in Figure \ref{cl-cyc2} then it will never come out of it.
The closed cycle $C=\{(a_2,b_2),(a_2,b_3),(a_3,b_3),(a_3,b_2)\}$ is a recurrent class and $(a_1,b_1)$ is an absorbing
state of Markov chain $P^0$ corresponding to the game given in Example \ref{ex1}.

We call a game acyclic if it has no closed cycles. The acyclic games include coordination games.
There  exists at least one SNE for acyclic games
from Theorem \ref{SNE_thm}.
For acyclic games the Markov chain defined by \eqref{TPM_Static_unperturb} is absorbing. Hence from the theory of Markov chain the
CBR dynamics given in Section \ref{Stoc_DP}
will be at SNE in the long run no matter from where it starts  \cite{Kemeny}.

\subsection{A stochastic CBR dynamics  with mistakes}\label{Stoc_per_DP}
Now, we assume that at each time $t$ the  selected coalition $S_t$ makes error in  making a deviation from $a^t$
and as a result it  moves to an action profile where some player(s)
in the coalition $S_t$ are not strictly better off.
We assume that at action profile $a^t$,
coalition $S_t$ makes error with $f(S_t,a^t)\varepsilon$ probability, where $f:\mathcal{S} \times A\rightarrow (0,\infty)$
and $0<~\varepsilon<~\frac{1}{M}$ with $M=\max_{S\in \mathcal{S},a\in A}f(S,a)$. The factor $f(S_t,a^t)$ shows the dependence of
coalition $S_t$ and
current action  profile $a^t$. The factor $\varepsilon$ determines the probability with which players in general make mistakes.
These mistakes add mutations to CBR dynamics and as a result we have
perturbed Markov chain $\{X_t^\varepsilon\}_{t=0}^\infty$ .
So, at time $t+1$ with probability $(1-f(S_t,a^t)\varepsilon)p_{\mathcal{I}_1}(a^{t+1}|S_t,a^t)$
the perturbed Markov chain switches to $a^{t+1}\in \mathcal{I}_1(S_t,a^t)$  and with
probability $f(S_t,a^t)\varepsilon  p_{\overline{\mathcal{I}}_1}(a^{t+1}|S_t,a^t)$
it switches to
$a^{t+1}\in \overline{\mathcal{I}}_1(S_t,a^t)$ ; $p_{\overline{\mathcal{I}}_1}(\cdot|S_t,a^t)$ is a probability measure over
finite set $\overline{\mathcal{I}}_1(S_t,a^t)$.
In the situation where there are no improved action profiles for coalition $S_t$, then $a^{t+1}=a^t$
with probability $1-f(S_t,a^t)\varepsilon$ and  $a^{t+1}\in \overline{\mathcal{I}}_1(S_t,a^t)\setminus \{a^t\}$
with probability $f(S_t,a^t)\varepsilon  p_{\overline{\mathcal{I}}_1\setminus \{a^t\}}(a^{t+1}|S_t,a^t)$;
$p_{\overline{\mathcal{I}}_1\setminus \{a^t\}}(\cdot|S_t,a^t)$ is a probability measure over finite set
$\overline{\mathcal{I}}_1(S_t,a^t)\setminus \{a^t\}$.
The transition law $P^\varepsilon$ of perturbed Markov chain is defined as below:

\begin{align}\label{TPM_Static}
 P^\varepsilon(X_{t+1}^\varepsilon=a'|X_{t+1}^\varepsilon=a)&=
 \sum_{S\in \mathcal{S}; \mathcal{I}_1(S,a)\neq \phi}p_S \big((1-f(S,a)\varepsilon)
 p_{\mathcal{I}_1}(a'|S,a)1_{\mathcal{I}_1(S,a)}(a')\nonumber\\
&\hspace{3cm} +f(S,a)\varepsilon p_{\overline{\mathcal{I}}_1}(a'|S,a) 1_{\overline{\mathcal{I}}_1(S,a)}(a')\big)\nonumber\\
 &\hspace{1cm}+\sum_{S\in \mathcal{S}; \mathcal{I}_1(S,a)= \phi} p_{S}\big((1-f(S,a)\varepsilon)1_{\{a'=a\}}(a')\nonumber\\
 &\hspace{3cm}+f(S,a)\varepsilon p_{\overline{\mathcal{I}}_1\setminus \{a\}}(a'|S,a)
 1_{\overline{\mathcal{I}}_1(S,a)\setminus \{ a\}}(a')\big),
\end{align}
for all $a, a'\in A$.


Given all possible coalitional moves and nonzero mutations, it is possible to reach  one action profile from  another with
positive probability in one step. This implies that the perturbed Markov chain $\{X_t^\varepsilon\}_{t=0}^\infty$
is aperiodic and irreducible. Hence, there exists a unique stationary distribution $\mu^\varepsilon$ for perturbed
Markov chain. However, when $\varepsilon=0$, there can be several stationary distributions corresponding to different
SNEs or closed cycles. Such Markov chains are called singularly perturbed Markov chains \cite{AFH02,AFH13}.
We are interested in the action profiles to which stationary distribution $\mu^\varepsilon$ assigns positive probability as
$\varepsilon\rightarrow 0$. This leads to the definition of a stochastically stable action profile.

 \begin{definition}
 An action profile $a$ is  stochastically stable relative to process $P^\varepsilon$ if
 $\lim_{\varepsilon\rightarrow 0} \mu_a^\varepsilon >~0$.
\end{definition}
We recall few definitions from  \cite{Young}.
From \eqref{TPM_Static}, we have $P^\varepsilon(a'|a)>0$ for all $a,a'\in A$.  The one step resistance from an action profile
$a$ to an action profile $a'\neq a$
is defined as the minimum number of mistakes (mutations) that are required for the transition from $a$ to $a'\neq a$ and it is denoted
by $r(a,a')$.
From \eqref{TPM_Static} it is clear that the transition from
$a$ to $a'$ has the probability of order $\varepsilon$ if $a'\notin \mathcal{I}_1(S,a)$
for all $S$ and thus has resistance 1 and is of order 1 otherwise, so has resistance 0.
So, in our setting $r(a,a')\in \{0,1\}$ for all $a,a'\in A$.
A zero resistance between two action profiles corresponds to a transition with positive probability under $P^0$.
One can view the action profiles as the nodes of a directed graph that has no self loops and the weight
of a directed edge between two different nodes
is represented by one step resistance between them.
Since $P^\varepsilon$ is an irreducible Markov chain then there must exist at least one directed path between any two recurrent
classes $H_i$ and $H_j$ of $P^0$ which starts from $H_i$ and ends at $H_j$. The resistance of any path is defined as the sum of the
weights of the corresponding edges. The  resistance of
a path   which is minimum among all paths from $H_i$ to $H_j$ is called as resistance from $H_i$
to $H_j$ and it is denoted by $r_{ij}$. The resistance from
 any action profile $a^i\in H_i$ to any action profile $a^j\in H_j$ is $r_{ij}$
 because inside $H_i$ and $H_j$ action profiles are connected with
 a path of zero resistance. Here $r_{ij}=1$
 because given all possible coalitional deviations it is always possible to reach from an action profile that belongs to $H_i$ to an action
profile belonging to $H_j$ in exactly 1 mutation.

 Now we recall the definition  of stochastic potential of a recurrent class $H_i$ of $P^0$ from \cite{Young}.
 It can be computed by restricting to a reduced graph.
Construct a  graph $\mathcal{G}$ where total number of nodes are the number of recurrent classes of $P^0$(one action profile
from each recurrent class)
 and a directed edge from $a^i$ to $a^j$  is weighted by $r_{ij}$. That is, the resistance of a directed  edge from
 $a^i$ to $a^j$ is 1.
 Take a node $a^i\in \mathcal{G}$  and consider all the spanning trees such that from every node
$a^j\in \mathcal{G}$, $a^j\neq a^i$,
there is a unique path directed from $a^j$ to $a^i$. Such spanning trees are called as $a^i$-trees.
The resistance of an $a^i$-tree  is the sum of the resistances of its edges.
The stochastic potential of $a^i$ is the resistance of an $a^i$-tree  having minimum resistance among all $a^i$-trees.
 The stochastic potential of each node in $H_i$ is same  \cite{Young},
which is a stochastic potential of $H_i$.
Suppose there are $J$ number of recurrent classes of $P^0$, then, an $a^i$-tree will have $J-1$ number of edges
and the resistance of each edge is 1. So, the resistance of each $a^i$-tree is $J-1$. This implies that the stochastic potential
of recurrent class $H_i$ is  $J-1$ and this is true for all the recurrent classes. So, in our case the stochastic potential of
all the recurrent classes of $P^0$ is same.

\begin{theorem}\label{SS_static_games}
All strong Nash equilibria and closed cycles of an $n$-player finite strategic game are stochastically stable.
\end{theorem}
\begin{proof}
We know that the Markov chain $P^\varepsilon$ is aperiodic and irreducible.
From \eqref{TPM_Static_unperturb} and \eqref{TPM_Static} it is easy to see that
\[
 \lim_{\varepsilon\rightarrow 0} P^\varepsilon(a'|a)=P^0(a'|a), \ \forall \ a, a'\in A.
\]
From \eqref{TPM_Static} it is clear that,  if $P^\varepsilon(a'|a) > 0$ for some $\varepsilon \in (0, \varepsilon_0]$, then we have
$$0<\varepsilon^{-r(a,a')}P^\varepsilon(a'|a)<\infty.$$
Markov chain $P^\varepsilon$ satisfies all three required conditions of Theorem~4 in \cite{Young} from which
it follows that as $\varepsilon \rightarrow 0$, $\mu^\varepsilon$ converges to some stationary distribution $\mu^0$
of $P^0$ and an action profile $a$ is stochastically stable, i.e., $\mu_a^0>0$ if and only if $a$ is contained
in a recurrent class of $P^0$ having minimum stochastic potential.
We know that the recurrent classes of Markov chain $P^0$ are strong Nash equilibria or closed cycles and the stochastic potential of all the recurrent classes are same.
Thus, all the strong Nash equilibria and closed cycles are  stochastically stable.
\end{proof}

\begin{remark}
Since the perturbed process $P^\varepsilon$ satisfies the conditions of Theorem~4 in \cite{Young}
for all functions $f(\cdot)$, the stochastic stability of strong Nash equilibria and closed cycles
is independent of $f(\cdot)$.
\end{remark}

We can have a similar
 CBR dynamics   without mistakes  and with mistakes as given in Sections \ref{Stoc_DP} and \ref{Stoc_per_DP}   respectively, if
 for all $S\in \mathcal{S}$ and $a\in A$  the set of improved action profiles is $\mathcal{I}_2(S,a)$
 as defined by \eqref{Impr_SSNE}.  We have the following result.

\begin{theorem}\label{SSNE_static_games}
All strict strong Nash equilibria and closed cycles of an $n$-player finite strategic game
are stochastically stable under corresponding CBR dynamics.
\end{theorem}

\begin{proof}
The proof follows from the  similar arguments given in Theorem \ref{SS_static_games}.
\end{proof}

\subsubsection{Equilibrium selection in  coordination games}
First we consider a $2\times 2$ coordination game and discuss which Nash equilibrium is selected by CBR
 dynamics  in the long run when probability of making mistakes vanish.
We compare our equilibrium selection results in $2\times 2$ coordination games with existing results
from  \cite{KMR},\cite{Young}. Later we discuss  equilibrium selection results in general $m\times m$ symmetric coordination games.

Consider a  $2\times 2$ coordination game,
$$
 \begin{matrix}
 s_1 \ ~~~~~~~~~~~~~~s_2 \\
 \renewcommand\arraystretch{1.5}
 \begin{matrix} s_1\\ s_2 \end{matrix}
 \begin{pmatrix}
 (a_{11},b_{11}) ~& (a_{12},b_{12}) \\
  (a_{21},b_{21}) ~& (a_{22},b_{22})
 \end{pmatrix}
 \end{matrix},
$$
where $a_{jk},b_{jk}\in\mathbb{R}$, $j,k\in\{1,2\}$ and $a_{11}>a_{21}$, $b_{11}>b_{12}$, $a_{22}>a_{12}$, $b_{22}>b_{21}$.
$A_i=\{s_1,s_2\}$, $i=1,2$.
Here $(s_1,s_1)$ and $(s_2,s_2)$ are two strict
Nash equilibria. In this game there are two types of Nash equilibria one is payoff dominant and other one is risk dominant.
If $a_{11}> a_{22}$, $b_{11}>b_{22}$, then $(s_1,s_1)$ is payoff dominant and if $a_{11}< a_{22}$, $b_{11}<b_{22}$, then
$(s_2,s_2)$ is payoff dominant. In other cases payoff dominant Nash equilibrium does not  exist.
From \cite{Young}, define,
$$R_1=\min\left\{\frac{a_{11}-a_{21}}{a_{11}-a_{12}-a_{21}+a_{22}}, \frac{b_{11}-b_{12}}{b_{11}-b_{12}-b_{21}+b_{22}}\right\},$$
$$R_2=\min\left\{\frac{a_{22}-a_{12}}{a_{11}-a_{12}-a_{21}+a_{22}},\frac{b_{22}-b_{21}}{b_{11}-b_{12}-b_{21}+b_{22}}\right\}.$$
If $R_1>R_2$, then $(s_1,s_1)$ is risk dominant Nash equilibrium and if $R_2>R_1$, then $(s_2,s_2)$ is risk dominant Nash equilibrium.
A payoff dominant Nash equilibrium is always an SNE. Hence, CBR dynamics  always selects
 payoff dominant Nash equilibrium whenever it exists. When payoff dominant Nash equilibrium  does not exist then both
 the Nash equilibria are
 strong Nash equilibria and in that case CBR dynamics selects both the Nash equilibria.
 While the stochastic dynamics  by  Young \cite{Young} always selects a risk dominant
Nash equilibrium.

A $2\times 2$ symmetric coordination game is considered by Kandori et al. \cite{KMR}. For this game $a_{jk}=b_{kj}$, $j,k\in\{1,2\}$.
In this case there always exists a payoff dominant Nash equilibrium.
Hence CBR dynamics always selects payoff dominant Nash equilibrium which is the only SNE.
While the stochastic dynamics  by  Kandori et al. \cite{KMR} always selects
a risk dominant Nash equilibrium.
Among symmetric coordination games if  we go beyond $2\times 2$ matrix games the result  by Young \cite{Young} cannot be generalized, i.e.,
it need not select a risk dominant Nash equilibrium.
Consider an example of $3\times 3$ matrix game from \cite{Young},
$$
 \begin{matrix}
 s_1 \ ~~~~~~~s_2 \ ~~~~~~s_3\\
 \renewcommand\arraystretch{1.5}
 \begin{matrix} s_1\\ s_2 \\ s_3\end{matrix}
 \begin{pmatrix}
 (6,6) ~& (0,5) ~& (0,0) \\
  (5,0) ~& (7,7) ~& (5,5)\\
  (0,0) ~& (5,5) ~& (8,8)
 \end{pmatrix}
 \end{matrix}.
$$
Here $(s_1,s_1)$, $(s_2,s_2)$ and $(s_3,s_3)$ are three Nash equilibria.
The stochastic dynamics  by Young \cite{Young} selects $(s_2,s_2)$ that is not a risk dominant Nash equilibrium.
A Nash equilibrium of an $m\times m$
symmetric coordination game is risk dominant if it is risk dominant in all pairwise contest   \cite{Harsanyi}.
For above $3\times 3$ game,  the Nash equilibrium $(s_3,s_3)$ is a risk dominant  as well as a payoff dominant and also an
SNE.
Hence, CBR dynamics
selects $(s_3,s_3)$. In fact for all $m\times m$
symmetric coordination game, CBR dynamics always selects a payoff dominant Nash equilibrium because it
is an SNE.

\section{Application to network formation games}\label{net-games}
In this section we consider the network formation games,  see e.g., some recent books
\cite{Jackson2}, \cite{Dutta-Jack}, \cite{Demange-Wooders}.
In general, the networks which are stable against the deviation of all the coalitions are
called as strongly stable networks. In the literature, there are two definitions of strongly stable networks. The first definition is
due to Dutta and Mutuswami \cite{Dutta-Mutu} that is corresponding to SNE. The second definition is due to
Jackson and van den Nouweland
\cite{Jack-Nou} that is corresponding
SSNE. A strongly stable network according to the definition of \cite{Jack-Nou} is also strongly stable network
according to the definition of \cite{Dutta-Mutu}. The definition of a strongly stable network according to Jackson and van den Nouweland
\cite{Jack-Nou} are
more often considered in the literature. We also consider the strong stability of networks according to Jackson and van den Nouweland
 \cite{Jack-Nou}.
We discuss the dynamic formation of
networks over infinite horizon.
We apply the CBR dynamics corresponding to SSNE  to network formation games to discuss the
stochastic stability of networks.

\subsection{The model}
Let $N=\{1,2,\cdots,n\}$ be a finite set of players also called as nodes.  The players are
 connected through  undirected edges.
 An edge can be defined  as a subset of $N$ of size 2, e.g., $\{ij\}\subset N$
 defines an edge between player $i$ and player $j$.
 The collection of edges define a network.
 Let $G$ denotes a set of all networks  on $N$.
  For each $i\in N$, let $u_i:G\rightarrow \mathbb{R}$ be
 a payoff function of  player $i$, where $u_i(g)$ is a payoff of player $i$ at network $g$.

%
To reach from one network to another requires the addition of new links
or the destruction of existing links. It is always assumed in the literature that forming a new link  requires the
consent of both the players while a player can delete a link unilaterally.
The coalition formation in network formation games has also been considered in
the literature.
 Some players in a network can form a coalition and make a joint move
to another network by adding or severing some links, if new network
is at least as beneficial as the previous network
for all the players of coalition and at least one player is strictly benefited (see \cite{Jack-Nou}, \cite{Jack_Survey}).
We recall few definitions from \cite{Jack-Nou} describing the coalitional moves in network
formation games and the stability of networks against all possible coalitional deviations.
\begin{definition}\label{def1}
 A network $g'$ is obtainable from $g$ via deviation by a coalition $S\in \mathcal{S}$ as denoted by $g\rightarrow_S g'$, if
 \begin{enumerate}
  \item  $ij\in g'$ and $ij\notin g$ then $\{i,j\}\subset S$.
  \item  $ij\in g$ and $ij \notin g'$ then $\{i,j\}\cap S \neq \phi$.
 \end{enumerate}
\end{definition}
\noindent The first condition of the above definition
requires that a new link can be added only between the nodes which are the part of a coalition $S$
and the second condition  requires that at least one node of any deleted link has to be a part of a coalition $S$.
We denote $G(S,g)$ as a set of all networks which are obtainable from $g$ via deviation by $S$, i.e.,
$G(s,g)=\{g'|g\rightarrow_S g'\}$.
\begin{definition}
  A deviation by a coalition $S$ from a network $g$ to a network $g'$ is said to be improving
if
\begin{enumerate}
 \item $g\rightarrow_S g'$,
 \item  $u_i(g')\geq u_i(g), \ \forall  \ i\in S$ (with at least one strict inequality).
\end{enumerate}
\end{definition}

We denote $\mathcal{I}_2(S,g)$ as a set of all networks $g'$ which are obtainable from $g$ by an improving deviation of $S$, i.e.,
\begin{align*}
 \mathcal{I}_2(S,g)=\big\{g'| g\rightarrow_S g', u_i(g')\geq u_i(g), \ \forall \ i\in S,  u_j(g')>u_j(g) \ \mbox{for some} \ j\in S \big\}.
 \end{align*}
It is clear that $g\notin \mathcal{I}_2(S,g)$ for all $S$.
We denote $\overline{\mathcal{I}}_2(S,g)=G(S,g)\setminus  \mathcal{I}_2(S,g)$ as a set of all networks
which are obtainable from $g$ due to erroneous decisions of $S$. This set is always nonempty as
$g\in \overline{\mathcal{I}}_2(S,g)$ for all $S$.

\begin{definition}
 A network $g$ is said to be strongly stable
 if it is not possible for any coalition $S$ to make an improving deviation from network $g$ to
 some other network $g'$.
\end{definition}

A strongly stable network need not always exist and in that case there exists some set of networks lying
on a closed cycle and all the networks in a closed cycle can be reached from each other via an improving path.
An improving path and a closed cycle in network formation games can be defined similarly to Definitions
 \ref{impr-path} and \ref{CC}, respectively.
 \begin{theorem}
There exists at least  a strongly stable network or a closed cycle of networks.
\end{theorem}
\begin{proof}
 The proof follows from the similar arguments used in Theorem \ref{SNE_thm}.
\end{proof}

\subsection{Dynamic network formation}
The paper by  Jackson and Watts \cite{JW} is the first one to consider the
dynamic formation of networks. They considered the case where at each time only a pair of players form a coalition
and only a link between them can be altered.
We consider the situation where  at each time a subset of players form a coalition and deviate from a current network
to a new network if at new network the payoff of each player of the coalition is at least as much as at current network
and at least one player has strictly better payoff.
 This process continues over infinite horizon.
 A coalition can make
all possible changes in the network and as a result more than one link can be created or severed
at each time. So, we consider the following network formation rules by Jackson and van den Nouweland
 \cite{Jack-Nou} given below:
 \begin{itemize}
  \item Link addition is bilateral, i.e., forming a link between player $i$ and player $j$ requires the consent of both players.
  \item Link destruction is unilateral, i.e., severing a link between player $i$ and player $j$ requires that player $i$ or player $j$
  or both agree to sever the link.
  \item At a time more than one link can be created or severed by the players.
 \end{itemize}
The CBR dynamics corresponding to SSNE can be applied to dynamic network formation. That is,  at  time $t$ a network is $g_t$ and a
coalition $S_t$
is selected with probability $p_{S_t}>0$ and it makes an improving deviation to a new network that is at least as beneficial
as $g_t$ for all  players
of coalition $S_t$ and at least one player of $S_t$ is strictly benefited.
  So, at time $t+1$ network is
$g_{t+1}\in \mathcal{I}_2(S_t,g_t)$ with probability $p_{\mathcal{I}_2}(g_{t+1}|S_t,g_t)$. If an improving
deviation is not possible for selected coalition $S_t$, then $g_{t+1}=g_t$. The above process defines  a Markov chain
over state space $G$ and its transition probabilities can be defined similarly to \eqref{TPM_Static_unperturb}.
In general this Markov chain is multichain whose absorbing set is either a strongly stable network
or a recurrent class having more than
one network is a closed cycle of networks.
We can also assume that at each time selected coalition $S_t$ makes error with small probability $f(S_t,g_t)\varepsilon$.
That is, $g_{t+1}\in \mathcal{I}_2(S_t,g_t)$  with probability \mbox{$(1-f(S_t,g_t)\varepsilon) p_{\mathcal{I}_2}(g_{t+1}|S_t,g_t)$}
and $g_{t+1}\in \overline{\mathcal{I}}_2(S_t,g_t)$ with probability $f(S_t,g_t)\varepsilon p_{\overline{\mathcal{I}}_2}(g_{t+1}|S_t,g_t)$.
The transition
probabilities of the perturbed Markov chain can be defined similarly to \eqref{TPM_Static}.
The presence  of mutations makes the Markov chain ergodic for which there exists a unique stationary distribution.
We are interested in the stochastically stable networks, i.e.,
the networks to which positive probabilities are assigned by the stationary distribution as $\varepsilon\rightarrow 0$.
The stochastic stability analysis similar to the one given in Section~\ref{Stoc_per_DP} holds here.
Thus, we have the following result.

\begin{theorem}
All the strongly stable networks and closed cycles of a network formation game
with the corresponding CBR dynamics are stochastically stable.
\end{theorem}

\begin{proof}
The proof follows directly from Theorem \ref{SSNE_static_games}.
\end{proof}

\section{Conclusions}\label{conclusion}
We introduce coalition formation among players in an $n$-player strategic game over infinite horizon and propose
a CBR dynamics. The mutations are present in the dynamics due to erroneous decisions taken by the coalitions.
We prove that all strong Nash equilibria and closed cycles of action profiles are stochastically stable, i.e.,
they are selected by the CBR dynamics as mutations vanish. Similar development holds for strict strong Nash
equilibria. We applied CBR dynamics to network formation games and prove that all strongly stable networks
and closed cycles of networks are stochastically stable.

\section*{Acknowledgements}

This research was sponsored by the European Commission within the framework
of the CONGAS project FP7-ICT-2011-8-317672.

%

\bibliographystyle{plain}
\bibliography{BibliographyExample}
\newpage
\appendix
\section{Algorithm for computing strong Nash equilibrium}\label{algo}
We give a finite step algorithm that computes an SNE whenever it exists. If an SNE does not exist
then in finite number of  steps the algorithm 
 confirms that there is no SNE.  
 From the definition of SNE
an action profile $a$ is an SNE if there is no improved action profile $a'\neq a$ 
for any coalition $S\in \mathcal{S}$, i.e., $\mathcal{I}_1(S,a)=\phi$ for all $S\in \mathcal{S}$.

\begin{algorithm}[]
\caption{}\label{algo1}
\begin{algorithmic}[1]
\STATE Choose $a\in A$. \label{action_profile}

\STATE Choose $S\in \mathcal{S}$. \label{Coalition}
\STATE Choose $a'_S\in A_S$. \label{profile_Coaltion}

\IF{$u_i(a'_S,a_{-S})>u_i(a), \ \forall \ i\in S$} \label{payoff-cond}

\STATE $A=A\setminus\{a\}$.
\IF {$|A|=0$} 
\STATE Go to Step \ref{alglaststep} 
\ELSE
\STATE Go to Step \ref{action_profile}. 
\ENDIF

\ELSE
\STATE $A_S=A_S\setminus \{a'_S\}$
\IF{$|A_S|=0$}
\STATE $\mathcal{S}=\mathcal{S}\setminus S$.
\IF{$|\mathcal{S}|=0$} 
\STATE Go to Step \ref{SNE_step}.
\ELSE
\STATE Go to Step \ref{Coalition} 
\ENDIF
\ELSE 
\STATE Go to Step \ref{profile_Coaltion}
\ENDIF

\ENDIF
\STATE  Strong Nash equilibrium does not exist\label{alglaststep}
\STATE  $a$ is Strong Nash equilibrium \label{SNE_step}
\end{algorithmic}
\end{algorithm}
\noindent The  Algorithm  \ref{algo1} terminates in finite number of steps because $A$ and $\mathcal{S}$ are finite.
 If we replace  Step \ref{payoff-cond} of the Algorithm \ref{algo1}  by $u_i(a'_S,a_{-S})\ge u_i(a), \ \forall \ i\in S$ together with
 at least one strict inequality, then Algorithm \ref{algo1} computes an SSNE whenever it exists.   
%
\newpage
\tableofcontents
\end{document}